\documentclass{article}

\PassOptionsToPackage{numbers}{natbib}

\usepackage[preprint]{neurips_2024}%
\usepackage[utf8]{inputenc} 
\usepackage[T1]{fontenc}    
\usepackage{hyperref}       
\usepackage{url}            
\usepackage{booktabs}       
\usepackage{amsfonts}       
\usepackage{nicefrac}       
\usepackage{microtype}      
\usepackage{xcolor}         
\usepackage{verbatim}
\usepackage{mathrsfs}
\usepackage{amsmath}
\usepackage{amsthm}
\usepackage{graphicx}
\newtheorem{theorem}{Theorem}[section]
\newtheorem{lemma}{Lemma}[section]
\newtheorem{corollary}{Corollary}[section]

\newtheorem{definition}{Definition}[section]

\title{Resolving positive semi-definiteness in physics-informed kernels for scientific machine learning}
\author{J. Moser$^{1}$, C. Albert$^{1}$, S. Ranftl$^{2, 3, \ast}$\\
        (1) Institute of Theoretical \& Computational Physics, Graz University of Technology, Austria\\
        (2) Courant Institute, New York University, New
York, USA \\
(3) Division of Applied Mathematics, Brown University, USA\\
($\ast)$ Current address: School of Mechanical Engineering, Purdue University, USA. \\
Correspondence: j.moser@tugraz.at, sranftl@purdue.edu 
       }

\begin{document}

\maketitle

\begin{abstract}
Many modern machine learning models can be understood as kernel-based function-space models, including Gaussian processes and neural tangent kernels. In scientific machine learning, differential operators are increasingly used to encode physical structure directly into such models. However, it has remained unclear under which conditions these constructions are valid machine learning models, i.e. preserve positive semi-definiteness, and whether observed instabilities arise from ill-posed modeling or numerical effects.
Here, we establish a simple and sufficient condition for positive semi-definiteness: for linear differential operators of order m, the base kernel must be m-times continuously differentiable. Crucially, this guarantee holds for operators with non-constant and even discontinuous coefficients. Examples are ubiquitous in physical systems, including diffusion, material elasticity, wave propagation in inhomogeneous media, and quantum systems.
We conclude that remaining instabilities are attributable to numerical issues, providing a unifying validity foundation for operator-informed kernel methods.
\end{abstract}

\begin{figure}[ht!]
\begin{center}
  \includegraphics[width=0.9\textwidth]{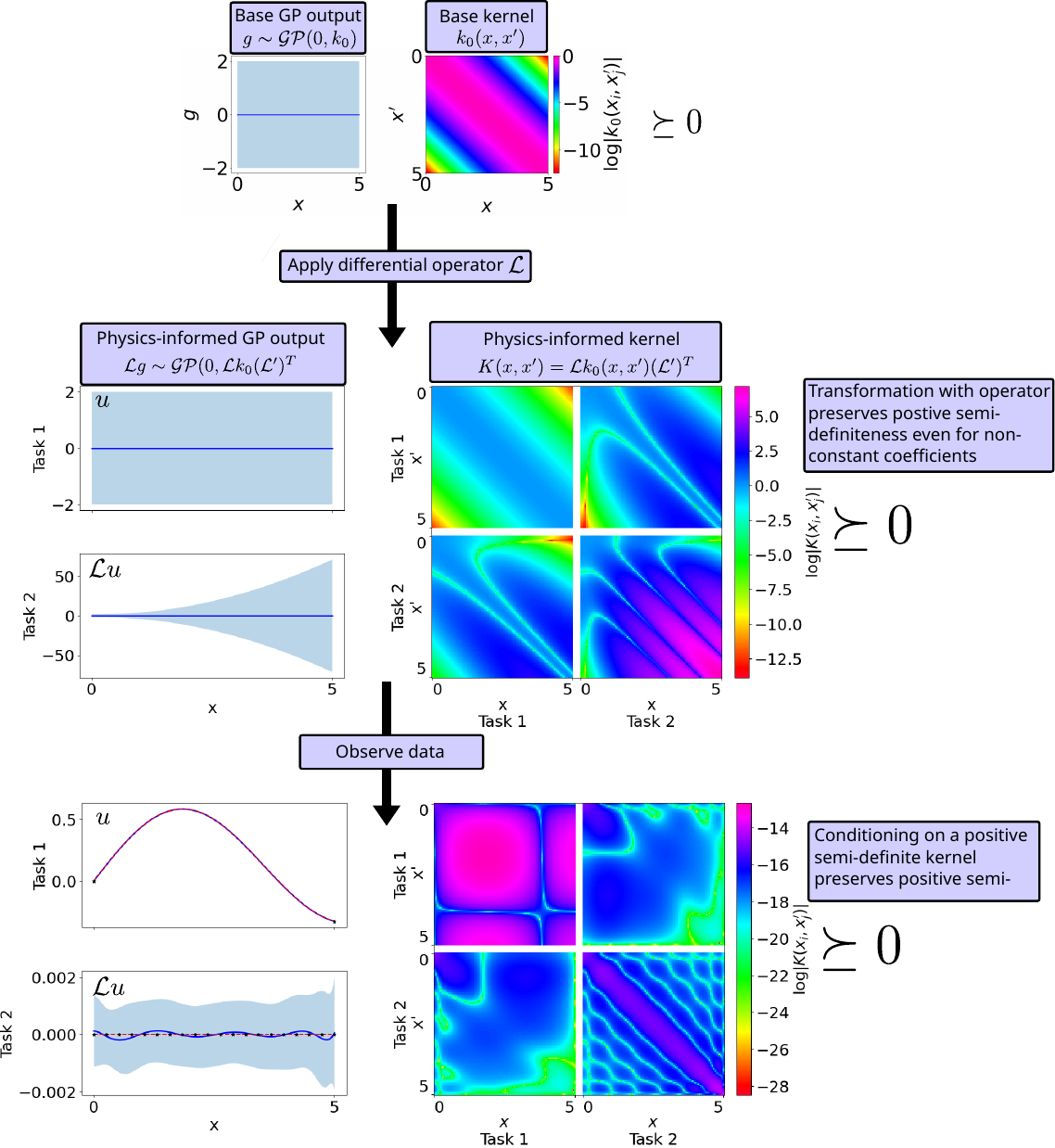}
\end{center}
\caption{Graphical abstract: Stepwise construction of a physics-informed Gaussian Process (PIGP) solving a Bessel equation. We show the PIGP output in the left panels (Task 1 and Task 2), characterized by the mean function in dark blue and variance through the $2\sigma$ interval in light blue, the true solution in red, and the corresponding kernel for the evaluated points as heatmaps in the right panels. Because there is a substantial scale difference between the kernel values across tasks, we show the logarithm of the absolute values of the kernel. We start with a simple zero-mean GP with a simple base kernel that is known to be positive semi-definite (top row). Application of the transformation (i.e. differential operator $\mathcal{L}$) results in a non-stationary, correlated multitask PIGP, with a highly non-trivial kernel structure (middle row). Our results guarantees the positive semi-definiteness of this physics-informed kernel. Lastly, we obtain the desired solution of the differential equation by conditioning on observed data (bottom row), which is known to preserve positive semi-definiteness.}
\label{fig:graphical_abstract}
\end{figure}

\section{Introduction}

Scientific machine learning seeks to incorporate physical structure into data-driven models, enabling learning systems that are consistent with governing laws \cite{karniadakis2021physics}. A common strategy is to encode differential equations directly into the model, either through loss functions as in physics-informed neural networks (PINNs) \cite{raissi2018physics}, or through kernel constructions in function space. In the latter case, physical structure is introduced by applying linear differential operators associated with ordinary or partial differential equations to kernel functions \cite{raissi_machine_2017, harkonen_gaussian_2023, jidling2017linearly}.

This operator-based perspective spans several model classes. In Gaussian processes, this leads to physics-informed Gaussian processes (PIGPs), where the kernel is transformed by the differential operator \cite{wahlstrom2015, raissi_machine_2017, lange-hegermann_algorithmic_2018, jidling2017linearly, macedo2010learning, pmm, besginow_constraining_2022} and has even been demonstrated for non-linear partial differential equations \cite{chen2021solving}. In neural networks, the neural tangent kernel (NTK) provides a complementary function-space view in the infinite-width limit \cite{jacot2018neural}, and has emerged as a fundamental tool for analyzing the training dynamics and failure modes of PINNs \cite{wang2022and}. These connections show that a wide range of modern learning methods can be understood as kernel-based models in which physical structure is imposed through differential operators.

Despite their increasing use, the theoretical validity of operator-informed kernel constructions remains a topic of discussion \cite{matsumoto, harkonen_gaussian_2023}. Particularly for physical systems where the governing differential operators have non-constant or discontinuous coefficients, users frequently observe numerical instability issues, and it remains unclear whether this behaviour stem from an invalid kernel construction or from numerical conditioning. This begs the question under which conditions the application of such operators preserves positive semi-definiteness of the kernel, a necessary and sufficient requirement for a valid kernel-based model. 

\subsection{Related work and contribution}

While the validity of simple differential kernels is well established \cite{rasmussen_gaussian_2006}, existing treatments of more general operator-informed kernels either sidestep the question of non-constant coefficients entirely or address it only under restrictive assumptions:
\citet{jidling2017linearly} require the relevant derivatives of the kernel to exist, but neither prove this nor discuss coefficient functions. \citet{lange-hegermann_algorithmic_2018} states closure of the kernel class only for constant-coefficient operators. \citet{harkonen_gaussian_2023} and \citet{pfortner2024} give measure-theoretic proofs of validity for linear continuous operators, and \citet{matsumoto} flag boundedness of differential operators as an issue that they resolve by providing abstract conditions based on dual spaces and integrability conditions on the operator. 
Therefore, to date, determining whether a given physics-informed kernel is actually admissible requires heavy measure-theoretic and functional analysis of the operator and the spaces it is defined on.

As a result, the practically important case of non-constant, and possibly discontinuous, coefficient functions is either excluded outright or tacitly assumed in the literature. 
Alas, differential operators involving non-constant or discontinuous coefficients are ubiquitous in the descriptions of physical systems. Examples include heterogeneous diffusion, elasticity of composites, wave propagation in inhomogeneous media, subsurface flow in porous materials, robotics \cite{solin2018modeling} and linear quantum systems such as the Schrödinger equation with spatially varying potential (e.g. quantum wells, step potentials).

In this work, we establish a simple and general sufficient condition for the validity of operator-informed kernels. For linear differential operators of order $m$, it is sufficient that the base kernel is $m$-times continuously differentiable and that the operator is linear in the state variable. Crucially, this result holds for operators with non-constant and even discontinuous coefficients provided the coefficients are well defined and finite at the evaluation points, thereby covering a broad class of physically relevant models.

As a consequence, outside this condition, positive semi-definiteness is no longer guaranteed, while compliance ensures that remaining instabilities are attributable solely to numerical issues, with spectral ill-conditioning being an important mechanism.  Our result therefore separates model validity from numerical stability and provides a unifying validity foundation for operator-informed kernel methods in scientific machine learning. Because our theorem applies to any positive semi-definite base kernel, 
it extends beyond Gaussian processes to neural tangent kernels and a broad class of function-space models, including variants such as Kriging and deep Gaussian processes \cite{damianou2013deep}, kernel ridge regression \cite{vovk2013kernel} and support vector machines \cite{suthaharan2016support}, random feature approximations \cite{rahimi2007random}, probabilistic numerics \cite{pmm}, and radial basis function interpolation methods.

\section{Notation and Definitions}
 \newcommand{\vv}{\mathbf}
 \newcommand{\op}{\mathcal{L}}
\newcommand{\fli}{\tilde{\op}_{im}\big\rvert_{\vv x = \vv x_i}}
\newcommand{\flj}{\tilde{\op}'_{jm}\big\rvert_{\vv x' = \vv x_j}}

Let $\mathcal{X} \subset \mathbb{R}^d$ be open and 
$k:\mathcal{X}\times\mathcal{X}\to\mathbb{R}$ a kernel.

\begin{definition} \label{def:PSD}
A kernel $k$ is positive semi-definite (PSD), denoted $k \succeq 0$, if
\begin{equation}
    \sum_{i,j=1}^N c_i k(\mathbf{x}_i,\mathbf{x}_j)c_j \geq 0 \label{eq:PSD1}
\end{equation}
for all $\mathbf{x}_1,\dots,\mathbf{x}_N \in \mathcal{X}$ and $c_1,\dots,c_N \in \mathbb{R}$.
\end{definition}

Any PSD kernel $k_0$ induces a reproducing kernel Hilbert space (RKHS) \cite{rkhs} $\mathcal{H}_0$ with inner product $\langle \cdot,\cdot\rangle_{\mathcal{H}_0}$, satisfying the \emph{reproducing property} 
\begin{equation} \label{eq:reproducing-property}
    \langle k_0(\mathbf{x},\mathbf{x}'), h \rangle_{\mathcal{H}_0} = h(\mathbf{x}')
    \quad \forall h \in \mathcal{H}_0.
\end{equation}

We consider linear differential operators $\mathcal{L}$ of order $m$, to be defined precisely in the main theorem.

Given a base kernel $k_0$, we define the operator-transformed kernel
\begin{equation}
    K(\mathbf{x},\mathbf{x}') = \mathcal{L} \, k_0(\mathbf{x},\mathbf{x}') \, (\mathcal{L}')^{\mathrm{T}},
    \label{eq:LkL}
\end{equation}
where $\mathcal{L}$ acts on $\mathbf{x}$ and the matrix transpose $(\mathcal{L}')^{\mathrm{T}}$ acts on $\mathbf{x}'$.

\begin{definition}[{\cite[Def.~4.36]{svm}}]
    Let $k_0(\vv x, \vv x')$ be a kernel on an open ${\cal X} \subset \mathbb{R}^d.$ For $m\geq 0$, we say that $k_0(\vv x, \vv x')$ is m-times \emph{continuously differentiable} if ${\cal D}^{\alpha}{\cal D}'^{\alpha}k_0(\vv x, \vv x'):{\cal X} \times {\cal X} \rightarrow \mathbb{R}$ exists and is continuous for all multi-indexes $\alpha \in \mathbb{N}_0^d$ with $|\alpha|\leq m$.
\end{definition}
Operator application to kernels of the form \eqref{eq:LkL} is understood in the classical sense, requiring the kernel to be continuously differentiable.

\section{Main Statement}


\begin{theorem}[Positive semi-definiteness of physics-informed kernels]

Let $\mathcal{X} \subset \mathbb{R}^d$ be a $d$-dimensional, open domain. 
Let $\mathcal{L}$ be a matrix of linear differential operators on $\mathcal{X}$ with coefficient functions $f_k$ defined on the whole domain. The elements of this matrix can be written as a sum over the finite index set $\mathcal{K}$
\begin{equation} \label{eq:general-diff-operator}
    \op_{ls}= \sum_{k\in \mathcal{K}} f_k(\vv x){\cal D}^{\alpha_k} := \sum_{k\in \mathcal{K}} f_k(\vv x)\frac{\partial^{|\alpha_k|}}{\partial x_1^{\alpha_{k_1}}\partial x_2^{\alpha_{k_2}}...} \;,
\end{equation}
with multi-indices $\alpha_k$ and $|\alpha_k|\leq m$, where $m \in \mathbb{N}$ is  the highest order of differentiation in any entry $\mathcal{L}_{l s}$ . 
Let $k_0(\mathbf{x}, \mathbf{x}') \in C^{m,m}(\mathcal{X}\times\mathcal{X})$ be a positive semi-definite base kernel function that is $m$-times continuously differentiable. 
Then, physics-informed kernels of the form 
\begin{equation}\notag
    K(\vv x, \vv x^\prime) =  \op k_0(\vv x,\vv x')(\mathcal{L}')^{\mathrm{T}}
\end{equation}
are positive semi-definite.
\label{th:1}
\end{theorem}

\textbf{Proof sketch}: The operator-transformed kernel can be written as a finite sum of compositions of bounded linear functionals acting on the base kernel. By the Riesz representation theorem in the reproducing kernel Hilbert space of the base kernel, these functionals admit representers, allowing the transformed kernel to be expressed as a sum of inner products. Positive semi-definiteness then follows from the non-negativity of the induced quadratic form. The details are shown in Sec. \ref{sec:proof}.

Importantly, the coefficients $f_k(\vv x)$ in \eqref{eq:general-diff-operator} need not  be constant nor continuous. The generality of the result stems from the simplicity of the proof, which relies only on fundamental properties of reproducing kernel Hilbert spaces, that is the reproducing property of the kernel and the Riesz representation theorem. These properties are also the foundations underlying commonly used concepts such as the kernel trick and the representer theorem. 


\section{Examples}
All experimental details for the following examples are documented in Appendix~\ref{app:experimental_details}.

\subsection{Physics}

Theorem \ref{th:1} implies the positive semi-definiteness of the physics-informed kernel and therefore validity of the corresponding physics-informed Gaussian process for several fundamental classes of physical systems ubiquitous in engineering and science.

Let's consider Bessel's equation. The solution for $n=1$ is visualized in the graphical abstract (Figure \ref{fig:graphical_abstract}). In operator-matrix structure, it is defined as
\begin{align}
   \vv x^2\partial_x^2y(\vv x)+\vv x\partial_xy(\vv x)+(\vv x^2-n^2)y(\vv x) = 0\,, 
\end{align}
and can be parametrized by a physics-informed kernel
\begin{align}
    K(\vv x, \vv x')&= \underbrace{\begin{pmatrix}1\\\vv x^2\partial_x^2+\vv x\partial_x+\vv x^2-n^2
    \end{pmatrix}}_{=\mathcal{L}} \underbrace{\begin{pmatrix}1&\vv x'^2\partial_x'^2+\vv x'\partial_x'+\vv x'^2-n^2
    \end{pmatrix}}_{=(\mathcal{L}^\prime)^{\mathrm{T}}}k_0(\vv x, \vv x')\;, 
    \label{eq:besselop}
\end{align}
where $\partial_x' \equiv \frac{\partial}{\partial \vv x'}$ denotes differentiation with respect to the second kernel argument, and, likewise, $\partial_x'^2$ denotes the second derivative w.r.t. the second kernel argument.  
The Schrödinger equation provides another example of a linear differential operator with spatially varying coefficients:
\begin{equation}
    \mathcal{L}\psi(\vv x)
    =
    -\frac{\hbar^2}{2m} \Delta \psi(\mathbf{x})
    +
    V(\mathbf{x}) \psi(\mathbf{x}), \label{eq:schroedinger}
\end{equation}
where $V(\mathbf{x})$ denotes the potential, which may be spatially varying or discontinuous.

Heterogeneous diffusion, heat conduction and subsurface flow in porous media are governed by operators of the form
\begin{equation} 
    \mathcal{L}u(\vv x)
    =
    -\nabla \cdot \big(\kappa(\mathbf{x}) \nabla u(\mathbf{x})\big), \label{eq:diffusion}
\end{equation}
where the conductivity $\kappa(\mathbf{x})$ may be spatially varying. 

Similarly, in linear elasticity of layered or composite materials, the displacement field $\mathbf{u}$ satisfies
\begin{equation} 
    \mathcal{L}\mathbf{u}(\vv x)
    =
    -\nabla \cdot \big(\mathbb{C}(\mathbf{x}) : \nabla^{s}\mathbf{u}(\mathbf{x})\big),
\end{equation}
where $\mathbb{C}(\mathbf{x})$ is the fourth-order stiffness tensor and
$\nabla^{s}\mathbf{u} = \frac{1}{2}(\nabla \mathbf{u} + \nabla \mathbf{u}^{T})$ is the symmetric gradient.

Both the diffusion and elasticity operators above are in divergence form,
$\mathcal{L}u(\vv x) = -\nabla \cdot (f(\mathbf{x}) \nabla u)$
(with $f = \kappa$ or $f = \mathbb{C}$, respectively).
When $f \in C^1$, the
product rule yields the canonical form (i.e. coefficient functions before operators)  for which Theorem \ref{th:1} is formulated
\begin{equation}
    -\nabla \cdot (f(\vv x) \nabla u(\vv x))
    =
    -f(\vv x) : \nabla^2 u(\vv x) - (\nabla \cdot f(\vv x)) : \nabla u(\vv x).
\end{equation}
Sharp interfaces where $f(\vv x)$ is discontinuous, as occur at material
boundaries, fall outside this direct expansion and would require a weak
formulation already at the problem-statement level, which lies out of the scope of Theorem~1.

Wave propagation in inhomogeneous media can be written, for example, as
\begin{equation} 
    \mathcal{L}u
    =
    \frac{1}{c(\mathbf{x})^2}\frac{\partial^2 u}{\partial t^2}
    -
    \Delta u,
\end{equation}
or, in frequency domain,
\begin{equation} 
    \mathcal{L}u
    =
    -\Delta u
    -
    \frac{\omega^2}{c(\mathbf{x})^2}u.
    \label{eq:wave}
\end{equation}
Here $c(\mathbf{x})$ denotes a spatially varying wave speed.
%
%

These differential equations are linear in the state variable but contain spatially varying or discontinuous coefficients. As long as the coefficients are well defined in the canonical formulation,  the examples fall directly within the scope of Theorem \ref{th:1}. 

To demonstrate the practical utility of our result, we showcase three numerical examples of operator-transformed kernels and their eigenvalues, in descending order,  in Figure \ref{fig:figure2}. In each example, the smallest eigenvalues, plotted in absolute value and therefore identifiable by the upturn at the end of the spectrum, are negative. That is, numerically, the physics-informed kernels appear to violate the PSD requirement -- within machine precision -- and would thus appear to constitute invalid models. Nevertheless, our result implies that this violation is only an artifact of numerical instability and therefore dependent on the particular discretisation. This qualitative behavior is robust to the choice of kernel length scale, as demonstrated by the corresponding parameter study in Appendix \ref{app:experimental_details} (Figure~\ref{fig:figure3}).

\begin{figure}[ht!]
\begin{center}
  \includegraphics[width=0.99\textwidth]{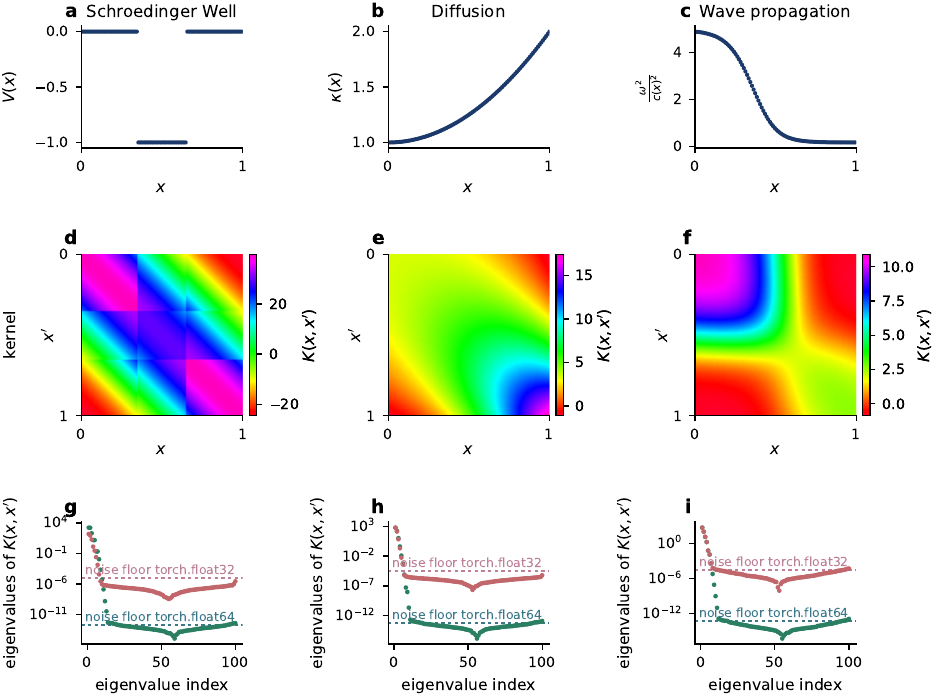}
\end{center}
\caption{Examples of operator-transformed kernels that fall within the scope of Theorem \ref{th:1}. In the columns, we show a Schrödinger equation with a well potential $V(x)$ (left), a diffusion problem with a quadratic diffusion coefficient $\kappa(x)$ (middle), and a wave propagation problem in a spatially varying medium with $c(x) = \tanh(5(x-0.5))+1.5$ (right). We show the spatially varying coefficients in the differential equations of each example in the top row, the values of the resulting operator-transformed kernel matrices evaluated on a rule of 100 collocation points in the middle row, and the kernel matrices' eigenvalues in the bottom row in descending order for 32- and 64-bit floating-point precision. For the smallest (negative) eigenvalues, we plot their absolute value to facilitate a logarithmic scale. The precision of the eigenvalue calculations is bounded by the unit roundoff $u$ (= half of the machine-epsilon) and the largest eigenvalue $\lambda_{max}$ and is marked by dashed lines. The obtained negative eigenvalues are in absolute value smaller than this precision, and hence illustrate the limitations of the current practice of numerically certifying positive semi-definiteness in floating-point arithmetic. This is now guaranteed by our Theorem \ref{th:1}.}
\label{fig:figure2}
\end{figure}

\subsection{Kernel-based machine learning models}
\subsubsection{Gaussian processes}
A Gaussian process (GP) $g \sim \mathcal{GP}(m,k)$ is fully specified by a mean function $m(\mathbf{x})$ and a covariance kernel $k(\mathbf{x},\mathbf{x}')$. For any finite set of inputs, the evaluations of $g$ follow a multivariate Gaussian distribution with covariance matrix $\Sigma_{ij} = k(\mathbf{x}_i,\mathbf{x}_j)$ for pairs $\mathbf{x}_i, \mathbf{x}_j \in \mathcal{X}$. By corollary of the Kolmogorov extension theorem, a Gaussian Process is valid iff the covariance kernel is symmetric and positive semi-definite.

Physics-informed Gaussian processes \cite{raissi_machine_2017, lange-hegermann_algorithmic_2018} incorporate differential equations by applying linear operators to a base kernel $k_0$. This yields operator-transformed kernels of the form \eqref{eq:LkL} in that 
\begin{equation}
\mathcal{L} g \sim \mathcal{GP}\big(\mathcal{L}m, \mathcal{L}k (\mathcal{L}')^{\mathrm{T}}\big)\;,
\end{equation}

i.e. applying a linear differential operator to a GP results again in a GP with operator-transformed kernel.
%
Kernel symmetry is implied by the symmetry of the base kernel and Schwarz's theorem for commuting differential operators. The validity of such models therefore hinges on the positive semi-definiteness of the transformed kernel $K$, directly linking to the main result of this work.

\subsubsection{Neural tangent kernels and physics-informed neural networks}

Neural tangent kernels (NTKs) provide a function-space characterization of neural networks in the infinite-width limit \cite{jacot2018neural}. In this regime, training a neural network by gradient descent is equivalent to kernel regression with a deterministic kernel $\Theta(\mathbf{x},\mathbf{x}')$, the NTK, which is positive semi-definite by construction.

In physics-informed neural networks (PINNs) \cite{raissi2018physics}, differential equations are enforced through the loss function by applying linear differential operators to the network output. In the NTK regime, this induces an operator transformation of the kernel, yielding an effective kernel of the form
\begin{equation} \notag
    K(\mathbf{x},\mathbf{x}') = \mathcal{L} \, \Theta(\mathbf{x},\mathbf{x}') \, (\mathcal{L}')^{\mathrm{T}},
\end{equation}
where $\mathcal{L}$ denotes the differential operator associated with the governing equation.

As in the Gaussian process setting, our theorem guarantees the positive semi-definiteness of the operator-transformed kernel $K$ arising in the NTK regime of physics-informed neural networks. The main result of this work therefore applies directly, providing a sufficient condition under which physics-informed neural tangent kernels are valid, even in the presence of spatially varying or discontinuous coefficients.

\section{Implications}

Our result establishes a simple and easily verifiable sufficient condition for positive semi-definiteness of physics-informed kernels (Gaussian processes, neural tangent kernels, etc.) associated with linear differential operators, including those with non-constant and discontinuous coefficients. Such operators arise widely in engineering and physics, including heat conduction, composite elasticity, wave propagation in inhomogeneous media and quantum mechanics. Because our results are formulated for matrix-valued operators, it also includes  differential equations with inhomogeneities, since unknown variables and inhomogeneities can  be modeled jointly in a multitask setting \cite{raissi_machine_2017}. The result therefore removes a key theoretical obstacle for applying physics-informed kernel models in these settings. 

By guaranteeing positive semi-definiteness under broad conditions, the result resolves the ambiguity between ill-posedness and ill-conditioning, or model validity and numerical stability. In particular, it proves that instabilities observed in practice are not due to ill-posedness of the model, but arise from spectral conditioning or other numerical issues of the resulting covariance matrices. Exact zero eigenvalues may occur by construction in multitask settings or when coefficient functions vanish, but are structurally benign. In contrast, near-zero eigenvalues dominate the condition number and are the primary source of numerical instability. These arise naturally, for example, when higher order differentials or coefficients varying sharply in space induce strong scale separation. 

From a practical perspective, such conditioning effects are well understood and can be mitigated using standard techniques, including appropriate scaling, regularization, and numerical stabilization. In Gaussian processes, adding jitter is common and affords the probabilistic interpretation of observing noisy data. The present result therefore separates a conceptual question of model validity and theoretical soundness from a numerical question of implementation.
Our proof is formulated in a Hilbert space setting and thus applies to essentially all standard base kernels used in practice, including the squared-exponential, polynomial, periodic or Matérn kernel. 

Extensions to weak formulations, including kernels with limited smoothness or more general functional settings such as Banach or Fréchet spaces \cite{matsumoto, harkonen_gaussian_2023}, remain an interesting direction for future work. Within the class of strong linear differential operators, however, the result provides a clean and practically relevant foundation for physics-informed kernel models.
Beyond GPs and NTKs,   also attention mechanisms -- as in transformer architectures -- have recently been understood as being essentially kernel-based mechanisms \cite{santos2026sparse}. Future work may include extensions or applications of our theorem in the context of transformer-based neural operator models.

\section{Proof} \label{sec:proof}
Let ${\cal H}_0$ denote the reproducing kernel Hilbert space (RKHS) of the base kernel $k_0(\vv x, \vv x')$.
We recall that a kernel is positive semi-definite if Definition \ref{def:PSD} is fulfilled.
%
%
We consider multi-output (matrix-valued) kernels $K:\mathcal{X}\times\mathcal{X}\to\mathbb{R}^{p\times p}$, where $[K(\mathbf{x},\mathbf{x}')]_{ls}$ denotes the covariance between outputs $l$ and $s$. An element of an operator-transformed kernel as defined in Eq. \eqref{eq:LkL} can explicitly be written as

 \begin{equation}
     [K(\mathbf{x},\mathbf{x}')]_{ls} = \Big[\op k_0(\vv x_i, \vv x_j)(\op')^T\Big]_{ls} := \sum_{m = 1}^q \op_{lm}\Big[\big[\op_{sm}'k_0(\vv x, \vv x') \big]\big\rvert_{\vv x' = \vv x_j}  \Big]\Big\rvert_{\vv x = \vv x_i} \;,
 \end{equation}
 where  $\big\rvert_{\vv x = \vv x_i}$ denotes evaluations of the expression in the bracket at a point $\vv x_i$ after operator application.

This covariance $K$ of a multi-output GP is a $p \times p$ matrix for each pair of pivot points, and one can easily verify that the structure of Eq.~\ref{eq:PSD1} can be split into outer and inner summations of the form
\begin{equation}
    \sum_{i,j = 1}^N \sum_{l,s = 1}^p c_{il}\Big[K(\vv x_i, \vv x_j)\Big]_{ls}c_{js} \geq 0
    \label{eq:PSD}
\end{equation}
which again needs to hold for all $\vv x_1, ....\vv x_N \in \mathcal{X}, c_{11}, c_{12} ... c_{Np} \in \mathbb{R}$, $N \in \mathbb{N}$.\\
In order to show positive semi-definiteness, we start by expressing the inner sums for arbitrary but fixed $c_{il}, c_{js}, \vv x_i, \vv x_j$: 
\begin{align}
    \sum_{l, s = 1}^p c_{il}c_{js}\big[K(\vv x_i,\vv x_j)\big]_{ls} 
    = \sum_{l, s = 1}^p c_{il}c_{js}\sum_{m = 1}^q \op_{lm}\Big[\big[\op_{sm}'k_0(\vv x, \vv x') \big]\Big\rvert_{\vv x' = \vv x_j}  \Big]\Big\rvert_{\vv x = \vv x_i} 
    \label{eq:cKc}
\end{align}
Changing the order of summation 
results in 
\begin{align}
    \sum_{m = 1}^q\Bigg( \sum_{l = 1}^p c_{il}\op_{lm} \bigg[ \sum_{s = 1}^p c_{js}\op'_{sm}\big[ k_0(\vv x, \vv x') \big]\big\rvert_{\vv x' = \vv x_j}  \bigg]\bigg\rvert_{\vv x = \vv x_i}\Bigg)
    =: \sum_{m = 1}^q \tilde{\op}_{im} \bigg[\tilde{\op}'_{jm}\big[ k_0(\vv x, \vv x') \big]\big\rvert_{\vv x' = \vv x_j}  \bigg]\bigg\rvert_{\vv x = \vv x_i}
    \label{eq:TTk}
\end{align}
with new operators $\tilde{\op}_{im}$, and their corresponding functionals $\fli$ defined as 
\begin{align}
\tilde{\op}_{im}\big\rvert_{\vv x = \vv x_i}&:{\cal H}_0 \rightarrow \mathbb{R}&
    \tilde{\op}_{im}\big[\,f(\vv x)\big]\big\rvert_{\vv x = \vv x_i}&:= \sum_{l = 1}^pc_{il}\op_{lm}\big[\,f(\vv x)\big]\big\rvert_{\vv x = \vv x_i}
\end{align}
The functionals $\tilde{\op}_{jm}'\big\rvert_{\vv x' = \vv x_j}$ are defined analogously. 

\begin{lemma}$\fli$ and $\flj$ are bounded linear functionals on ${\cal H}_0$. 
\end{lemma}
\begin{proof}
    Since the operators $\op_{lm}$ are linear, so are the functionals $\fli$. 
    A functional $L:{\cal H}_0 \rightarrow \mathbb{R}$ is bounded if there exists a bound $0<M<\infty$ such that $|L [h (\vv x)]|\leq M||h(\vv x)||_{{\cal H}_0} \, $ for all $ h(\vv x) \in {{\cal H}_0}$.
    To find such a bound, we use the following corollary.
    \begin{corollary}[{\cite[Corollary~4.36]{svm}}]
        Let ${\cal X}\subset \mathbb{R}^d$ be an open subset, $m \geq 0,$ and $k_0(\vv x, \vv x')$ be an $m$-times continuously differentiable kernel on ${\cal X}$ with RKHS ${\cal H}_0$. Then every $h \in {\cal H}_0$ is $m$-times continuously differentiable, and for $\alpha \in \mathbb{N}_0^d$ with $|\alpha|\leq m$ and $x \in {\cal X}$ we have
   
\begin{align}
    |{\cal D}^\alpha h(\vv x)| \leq ||h(\vv x)||_{{\cal H}_0}\sqrt{{\cal D}^{\alpha}\mathcal{D}^{\alpha'}k_0(\vv x,\vv x)}.
    \label{eq:differentials}
\end{align}
 \end{corollary}
    So for the general form of differential operators $\op_{lm} = \sum_{k\in\mathcal{K}}f_k(\vv x){\cal D}^{\alpha_k}$, evaluated at the fixed point $\vv x_i$, 
    we find a bound for the functionals $\fli$:
\begin{align}
    &\big|\tilde{\op}_{im}\big[\,h(\vv x)\big]\big\rvert_{\vv x = \vv x_i}\big| \leq \sum_{l = 1}^p\big|c_{il}\big|\big|\op_{lm}\big[h(\vv x)\big]\big\rvert_{\vv x = \vv x_i}\big|
    \leq \sum_{l = 1}^p\sum_{k\in\mathcal{K}}\big|f_k(\vv x_i)c_{il}\big|\Big|{\cal D}^{\alpha_k}\big[h(\vv x)\big]\big\rvert_{\vv x = \vv x_i}\Big|\nonumber \\
    &\Bigg(\leq\sum_{l = 1}^p\sum_{k\in\mathcal{K}}\big|f_k(\vv x_i)c_{il}\big| \sqrt{{\cal D}^{\alpha_k}\mathcal{D'}^{\alpha_k}k_0(\vv x, \vv x')\big\rvert_{\vv x = \vv x' = \vv x_i}}\Bigg)||h(\vv x)||_{{\cal H}_0} =:M||h(\vv x)||_{{\cal H}_0}.
    \label{eq:boundedness2}
\end{align}
The analogous statement holds for $\flj$.
\end{proof}

\begin{theorem}[Riesz Representation Theorem\cite{riesz}]
Let $L$ be a bounded linear functional on a Hilbert space ${\cal H}$ (i.e. $L \in {\cal H}^*$). Then, every $L$ has a unique representer $\phi_L \in {\cal H}$, such that $L[h(\vv x)] = \langle h(\vv x), \phi_L(\vv x)\rangle_{\cal H}$ for all $h(\vv x) \in {\cal H}$.
\label{th:riesz}
\end{theorem}
The linearity and boundedness of $\fli$ and $\flj$ enables us to use Theorem \ref{th:riesz}, defining the representer of $\fli$ resp. $\flj$ as $\phi_{im}$ resp. $\phi_{jm}$. Combining this with the reproducing property (Eq.~\ref{eq:reproducing-property}) \cite{rasmussen_gaussian_2006} of the base kernel $k_0(\vv x, \vv x')$ in ${\cal H}_0$,
we can therefore write the action of our functionals on any of the arguments of the base kernel as
\begin{align}
    \tilde{\op}_{im}\big[k_0(\vv x, \vv x')\big]\big\rvert_{\vv x = \vv x_i} = \big\langle k_0(\vv x, \vv x'), \phi_{im}(\vv x)\big\rangle_{{\cal H}_0} = \phi_{im}(\vv x') \quad \forall \vv x' \in {\cal X}\\
    \tilde{\op}'_{jm}\big[k_0(\vv x, \vv x')\big]\big\rvert_{\vv x' = \vv x_j} = \big\langle k_0(\vv x, \vv x'), \phi_{jm}(\vv x')\big\rangle_{{\cal H}_0} = \phi_{jm}(\vv x) \quad \forall \vv x \in {\cal X}.
\end{align}
Equipped with this knowledge, we can further transform equation \ref{eq:cKc}
\begin{align}
    \sum_{l, s = 1}^p c_{il}c_{js}\big[K(x_i,x_j)\big]_{ls}
    &= \sum_{m = 1}^q \tilde{\op}_{im} \bigg[\tilde{\op}'_{jm}\big[ k_0(\vv x, \vv x') \big]\big\rvert_{\vv x' = \vv x_j}  \bigg]\bigg\rvert_{\vv x = \vv x_i}\\
    &= \sum_{m = 1}^q \tilde{\op}_{im} \Big[\big\langle k_0(\vv x, \vv x'), \phi_{jm}(\vv x')\big\rangle_{{\cal H}_0}\Big]\Big\rvert_{\vv x = \vv x_i}\nonumber \\
     &= \sum_{m = 1}^q \big\langle\tilde{\op}_{im} \big[k_0(\vv x, \vv x')\big]\big\rvert_{\vv x = \vv x_i}, \phi_{jm}(\vv x')\big\rangle_{{\cal H}_0}\nonumber
    = \sum_{m = 1}^q \big\langle \phi_{im}(\vv x'), \phi_{jm}(\vv x')\big\rangle_{{\cal H}_0}.
\end{align}

where we used the linearity and boundedness of $\fli$ in the second-to-last step.
Lastly, let's return to the original equation \ref{eq:PSD}:
\begin{align}
        \sum_{i,j = 1}^N \sum_{l,s = 1}^p c_{il}\big[K(\vv x_i, \vv x_j)\big]_{ls}c_{js} 
        &= \sum_{i,j = 1}^N \sum_{m = 1}^q \langle \phi_{im}(\vv x'), \phi_{jm}(\vv x')\rangle_{{\cal H}_0} \\
        &= \sum_{m = 1}^q\left\langle \sum_{i = 1}^N \phi_{im}(\vv x'), \sum_{j = 1}^N\phi_{jm}(\vv x')\right\rangle_{{\cal H}_0} \nonumber \\
    &= \sum_{m = 1}^q \left\Vert \sum_{i = 1}^N\phi_{im}(\vv x')\right\Vert^2_{{\cal H}_0}    
    \geq 0 \nonumber
    \label{eq:proof}
\end{align}
The second-to-last equality is due to the bilinearity of the inner product and the last step is a simple renaming of the summation index $j$ to $i$, which concludes our proof of positive semi-definiteness. \\
\rightline{\qedsymbol}

\section*{Acknowledgement}

J.M. and C.A. were financially supported by the Austrian Research Promotion Agency (FFG) project VENTUS within the AI for Green funding program, Grant no. 910263. S.R. was financially supported by the Austrian Science Fund (FWF), Grant DOI: 10.55776/J4774. This work has been carried out within the framework of the EUROfusion Consortium, funded by the European Union via the Euratom Research and Training Programme (Grant Agreement No 101052200 — EUROfusion). Views and opinions expressed are however those of the author(s) only and do not necessarily reflect those of the European Union or the European Commission. Neither the European Union nor the European Commission can be held responsible for them.

\bibliography{literature_psd}

\begin{appendix}
\section{Experimental Details}
\label{app:experimental_details}

This section provides the experimental details used for Figures ~\ref{fig:graphical_abstract} and~\ref{fig:figure2}. Code is available at \url{https://github.com/moserjo/PSD}.

All experiments use the squared exponential kernel
\[
k_0(\vv x, \vv x') = A\exp\left(-\frac{(\vv x - \vv x')^2}{2l}\right)
\]
as the base kernel, with amplitude $A$ and length scale $l$ as hyperparameters.

\paragraph{Bessel equation (Figure~\ref{fig:graphical_abstract}).}
The operator ${\cal L}$ defined in Equation \eqref{eq:besselop} is matrix valued and gives therefore rise to a multitask (= multi-output) kernel. The two output tasks are conditioned on a total of 22 pivot points: the first task corresponding to the solution $u$ uses two boundary conditions, $(0,0)$ and 
$(5, J_1(5))$ with $J_1(5) \approx -0.33$; the second task, corresponding to ${\cal L}u$, uses 20 linearly spaced pivot points in $[0, 5]$ with their function value constrained to zero. Each task is evaluated at 100 linearly spaced test points. Hyperparameters $A$ and $l$ are optimized for 500 iterations using ADAM with learning rate $0.1$ and convergence was confirmed manually.

\paragraph{Quantum well, diffusion, and wave equations (Figure~\ref{fig:figure2}).}
The kernels of the three examples were transformed using the operators of Equations~\eqref{eq:schroedinger}, \eqref{eq:diffusion}, \eqref{eq:wave} and use the coefficients
\begin{align}
    \notag
    V(x) &= \begin{cases}
        0 & 0.35 < x < 0.65 \\
        -1 & \text{else}
    \end{cases}\\
    \kappa(x) &= 1 + x^2\\
    c(x) &= 1.5 + \tanh(5(x - 0.75)), \quad \omega = 1.
\end{align}
No conditioning is applied here; we evaluate only the prior kernel at 100 linearly spaced points on $x \in [0,1]$, with hyperparameters fixed to $l = 1$ and $A = 1$. Eigenvalues in Figure \ref{fig:figure2} are computed both in single- and in double-precision floating point using the linalg.eigvalsh method from the \texttt{PyTorch} library. Due to the floating-point precision it is not sensible to give actual condition numbers. In all examples, between 6 to 20 out of 100 eigenvalues were above the noise floor.

\paragraph{Length scale dependence (Figure~\ref{fig:figure3}).}
To test whether the ill-conditioning is an artifact of the hyperparameter choice, we recompute the eigenvalues for $l \in \{0.1, 0.5, 1, 3, 10\}$ with 64-bit floating-point precision, keeping $A=1$ (since $A$ only rescales the eigenvalues). As shown in Figure~\ref{fig:figure3}, the  eigenvalues below noise floor shift slightly left and downwards with increasing length scale, but the qualitative behaviour persists: the kernels remain severely ill-conditioned regardless of length scale.

\begin{figure}[ht!]
\begin{center}
  \includegraphics[width=0.9\textwidth]{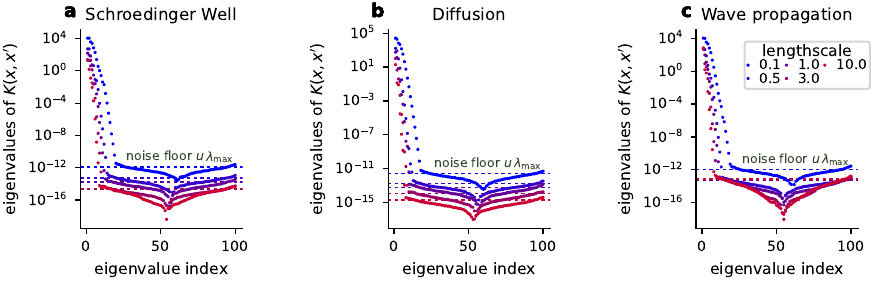}
\end{center}
\caption{Eigenvalues of the three examples from Figure~\ref{fig:figure2} for varying length scale $l \in \{0.1, 0.5, 1, 3, 10\}$. The ill-conditioning persists across all hyperparameter values.}
\label{fig:figure3}
\end{figure}

\end{appendix}

\end{document}